\documentclass[10pt,conference]{IEEEtran}

\usepackage{amsfonts}
\usepackage{amssymb}
\usepackage{mathrsfs}
\usepackage{verbatim}
\usepackage{times}
\usepackage{subfigure}
\usepackage[centertags]{amsmath}
\usepackage{graphicx}
\newtheorem{thm}{\bf Theorem}

\newtheorem{defn}{{Definition}}

\newtheorem{eg}{Example}[section]

\begin{document}

\title{Noncoherent Low-Decoding-Complexity Space-Time Codes for Wireless Relay Networks}

\author{
\authorblockN{G. Susinder Rajan}
\authorblockA{ECE Department \\
Indian Institute of Science \\
Bangalore 560012, India \\
susinder@ece.iisc.ernet.in}
\and
\authorblockN{B. Sundar Rajan}
\authorblockA{ECE Department \\
Indian Institute of Science \\
Bangalore 560012, India \\
bsrajan@ece.iisc.ernet.in}
}

\maketitle

\begin{abstract}
The differential encoding/decoding setup introduced by Kiran \emph{et al}, Oggier \emph{et al} and Jing \emph{et al} for wireless relay networks that use codebooks consisting of unitary matrices is extended to allow codebooks consisting of scaled unitary matrices. For such codebooks  to be used in the Jing-Hassibi protocol for cooperative diversity, the conditions that need to be satisfied by the relay matrices and the codebook are identified. A class of previously known rate one, full diversity, four-group encodable and four-group decodable Differential Space-Time Codes (DSTCs) is proposed for use as Distributed DSTCs (DDSTCs) in the proposed set up. To the best of our knowledge, this is the first known low decoding complexity DDSTC scheme for cooperative wireless networks.
\end{abstract}

\section{Introduction}
\label{sec1}
Recently \cite{JiH4,RaR1}, the idea of space-time coding for colocated MIMO channels has been applied in the setup of cooperative wireless networks in the name of distributing space-time coding, wherein coding is performed across users and time. This strategy provides each user a diversity order equal to the number of cooperating terminals even though all the users are only equipped with a single antenna. The diversity thus achieved is called as cooperative diversity. However, such strategies require that the destination have complete knowledge of the fading coefficients from all the users to itself as well as that of the fading coefficients between users. But, to obtain the knowledge of the fading coefficients between the users at the destination requires additional resources. To solve this problem, in \cite{KiR2}, Kiran \emph{et al} have proposed a differential encoding/decoding setup for cooperative wireless networks that does not require the knowledge of fading coefficients between the users at the destination. Such codes were named as partially coherent distributed space-time codes in \cite{KiR2}. However, in a recent work \cite{OgH2}, it has been shown that the same strategy of \cite{KiR2} offers full diversity with a suboptimal receiver that does not require the knowledge of any of the fading coefficients. In \cite{JiJ}, Jing \emph{et al} have proposed a differential encoding/decoding setup for cooperative wireless networks which is more general than the setup proposed in \cite{KiR2} and they have also provided few code constructions. We call the class of Differential Space-Time Codes (DSTCs) used in a distributed manner for cooperative diversity  as Distributed DSTCs (and denote by DDSTCs) to differentiate them from DSTCs for colocated MIMO systems. The problem of designing DDSTCs is more challenging than the DSTC design problem for colocated MIMO channels since in this scenario we have additional constraints to be satisfied which are due to the cooperative diversity protocol. However, both \cite{KiR2} as well as \cite{OgH2} do not address the problem of designing low encoding and low decoding complexity DDSTC. In \cite{JiJ}, a solitary example of the Alamouti code has been proposed but a general construction targeting low decoding complexity is not available. This issue gains significant importance especially if the number of cooperating terminals is large, which is quite expected in applications like wireless sensor networks. The contributions of this paper are as follows:
\begin{itemize}
\item The differential encoding/decoding setup introduced by Kiran \emph{et al} \cite{KiR2}, Oggier \emph{et al} \cite{OgH2} and Jing \emph{et al} \cite{JiJ} for wireless relay networks that use codebooks consisting of unitary matrices is extended to allow codebooks consisting of scaled unitary matrices.
\item When the codebook of scaled unitary matrices is obtained from a design with proper choice of signal sets for the variables and  used in the Jing-Hassibi protocol \cite{JiH4} for cooperative diversity, the conditions involving the relay matrices and the codebook that need to be satisfied are identified. 

\item In \cite{RaRisit1}, a class of rate one, full diversity, four-group encodable and four-group decodable DSTCs is constructed for $N_T=2^{\lambda}$ transmit antennas using matrix representation of extended Clifford algebras and appropriate choice of signal sets for colocated MIMO communication. We prove using algebraic techniques that these codes satisfy the conditions mentioned above and hence are usable as DDSTCs. 

\item To the best of our knowledge, this is the first known low decoding complexity DDSTC scheme for cooperative wireless networks.
\end{itemize}

The rest of the paper is organized as follows:  Section \ref{sec2} introduces the system model for the cooperative diversity scheme employing a differential scaled unitary matrix code book at the source. In Section \ref{sec3} the notion of $g$-group encoding cum decoding for the system model of Section \ref{sec2} is given and the problem of DDSTC design is formally presented. The extra conditions on the code structure imposed by the cooperative diversity protocol of \cite{KiR2}-\cite{OgH2} are then described. Also, we briefly describe the DSTCs from extended Clifford algebras developed in \cite{RaRisit3} for coherent colocated MIMO communication. In Section \ref{sec4} we show that the DSTCs from extended Clifford algebras of \cite{RaRisit3} satisfy the conditions needed for them to be a suitable code as DDSTCs for cooperative diversity systems with scaled unitary codebook. 

\section{System model and differential scaled unitary matrix codebook}
\label{sec2}  
In this section, we briefly explain the distributed differential encoding/decoding setup proposed in \cite{KiR2,JiJ,OgH2} with a slight modification. We then highlight the various requirements for this coding problem.

We consider a network consisting of a source node, a destination node and $R$ other relay nodes which aid the source in communicating information to the destination. All the nodes are assumed to be equipped only with a single antenna and are half duplex constrained, i.e., a node cannot transmit and receive simultaneously. The channel path gains from the source to the $i$th relay, $f_i$ and from the $j$th relay to the destination $g_j$ are all assumed to be independent and identically distributed complex Gaussian random variables with zero mean and unit variance. Moreover, we assume symbol synchronization among all the nodes.

Every transmission cycle from source to destination comprises of two stages. In the first stage, the source transmits a $T(T\geq R)$ length vector $\sqrt{\pi_1P}s$ which the relays receive. Here, $P$ denotes the total power spent by all the relays and the source and $\pi_1$ is the power allocation factor denoting the fraction of $P$ spent by the source. The received vector at the $j$th relay node is then given by  
\begin{equation}
r_j=\sqrt{\pi_1P}f_js+v_j,\ \mathrm{with}\ v_j\sim\mathcal{CN}(0,I_T).
\end{equation} 
In the second half of the cycle, all the relay nodes are scheduled to transmit together. The $j$th relay node transmits a $T$ length vector $t_j$ which is a function of $r_j$. The relays are only allowed to linearly process the received vector $r_j$ or its conjugate $r_j^*$. To be precise, the $j$th relay node is equipped with a $T\times T$ unitary matrix $A_j$ (called relay matrix) and it transmits  $t_j=\sqrt{\frac{\pi_2P}{\pi_1P+1}}A_jr_j$ or $t_j=\sqrt{\frac{\pi_2P}{\pi_1P+1}}A_jr_j^*$. Without loss of generality, we may assume that $M$ relays linearly process $r_j$ and the remaining $R-M$ relays linearly process $r_j^*$. The quantity $\pi_2$ is the power allocation factor denoting the fraction of $P$ spent by a relay. The vector received at the destination after $2T$ time slots is given by

\begin{equation}
y=\sum_{j=1}^{R}g_jt_j+w=\sqrt{\frac{\pi_1\pi_2P^2}{\pi_1P+1}}XH+N
\end{equation}
\noindent
where,\\ $H=\left[\begin{array}{ccccccc}f_1g_1 & f_2g_2 & \dots f_Mg_M & f_{M+1}^*g_{M+1} & \dots & f_Rg_R\end{array}\right]^T$,\\ $N=\sqrt{\frac{\pi_2P}{\pi_1P+1}}\left(\sum_{j=1}^{M}g_jA_jv_j+\sum_{j=M+1}^{R}g_jA_jv_j^*\right)+w$,\\ $X=\left[\begin{array}{ccccccc}A_1s & \dots & A_Ms & A_{M+1}s^* & \dots & A_Rs^*\end{array}\right]$ and $w\sim\mathcal{CN}(0,I_T)$.

The differential encoding is performed at the source as follows. A known vector $s_0$ is transmitted by the source in the first cycle. The transmitted vector at the $t$-th cycle is then given as follows
\begin{equation}
s_t=\frac{1}{a_{t-1}}U_ts_{t-1}
\end{equation}
\noindent
where, $U_t\in\mathscr{C}$ is the codeword containing the information at the $t$th cycle which satisfies $U_t^HU_t=a_t^2I_T,\ a_t\in\mathbb{R}$. The originally proposed coding strategies in \cite{KiR2,JiJ,OgH2} force $a_t=1$. The received vector at the destination in the $t$-th cycle can be written as
\begin{equation}
y_t=\sqrt{\frac{\pi_1\pi_2P^2}{\pi_1P+1}}X_tH_t+N_t
\end{equation}
\noindent
where,\\
$X_t=\left[\begin{array}{ccccccc}A_1s_t & A_2s_t & \dots & A_Ms_t & A_{M+1}s_t^* & \dots & A_Rs_t^*\end{array}\right]$. If
\begin{equation}
\begin{array}{rcl}
A_iU_t&=&U_tA_i,\ \forall\ U_t\in\mathscr{C}, i=1,\dots,M,\\
A_iU_t^*&=&U_tA_i,\ \forall\ U_t\in\mathscr{C}, i=M+1,\dots,R
\end{array}
\end{equation}
then, substituting for $s_t$ we have
\begin{equation}
\begin{array}{rcl}
X_t&=&\frac{1}{a_{t-1}}U_tX_{t-1}.
\end{array}
\end{equation}
If the channel remains approximately constant for more than $4T$ channel uses, then we can assume that $H_t=H_{t-1}$. Thus $y_t$ can be expressed as
\begin{equation}
\begin{array}{rl}
y_t&=\sqrt{\frac{\pi_1\pi_2P^2}{\pi_1P+1}}X_tH_t+N_t\\
&=\sqrt{\frac{\pi_1\pi_2P^2}{\pi_1P+1}}\frac{1}{a_{t-1}}U_tX_{t-1}H_{t-1}+N_t\\
&=\frac{1}{a_{t-1}}U_ty_{t-1}+\hat{N}_t
\end{array}
\end{equation}
where, $\hat{N}_t=-\frac{1}{a_{t-1}}U_tN_{t-1}+N_t$. Now, we can decode the codeword $U_t$ as follows
\begin{equation}
\label{eqn_metric0}
\hat{U_t}=\mathrm{arg} \min_{U_t\in\mathscr{C}}\parallel y_t-\frac{1}{a_{t-1}}U_ty_{t-1}\parallel^2
\end{equation}
where, $a_{t-1}$ can be estimated from the previous decision. Note that the above decoder is not a Maximum-Likelihood (ML) decoder. However, for the colocated MIMO case, it has been shown in \cite{TaC}-\cite{YGT2} that the performance loss is negligible. In this setup, we shall call $\mathscr{C}$ to be a DDSTC in order to distinguish it from DSTCs for colocated MIMO systems. We shall choose $\mathscr{C}$ to be a linear STBC in order to reap the benefits of low encoding and decoding complexity. However, the problem is not straightforward since we need to satisfy few more additional conditions. These are illustrated in the following subsection more precisely.

\section{Problem Statement}
\label{sec3}

\begin{defn} A linear design $S(x_1,x_2,\dots,x_K)$ in $K$ real indeterminates or variables $x_1,x_2,\dots,x_K$ is a $n\times n$ matrix with entries being a complex linear combination of the variables. More precisely, it can be written as follows,
$$
S(x_1,x_2,\dots,x_K)=\sum_{i=1}^{K}x_iB_i
$$
where, $B_i\in \mathbb C^{n\times n}$ are called the weight matrices. A linear STBC $\mathscr{C}$ is a finite set of $n\times n$ complex matrices which can be obtained by taking a linear design $S(x_1,x_2,\dots,x_K)$ and specifying a signal set $\mathscr{A}\subset\mathbb{R}^{K}$ from which the information vector $X=\left[\begin{array}{cccc}x_1 & x_2 & \dots & x_K\end{array}\right]^T$ take values from, with the additional condition that $S(a)\neq S(a'), \forall\ a\neq a'\in\mathscr{A}$. A linear STBC $\mathscr{C}=\left\{S(X)|X\in \mathscr{A}\right\}$ is said to be $g$-group encodable (or $\frac{K}{g}$ real symbol encodable or $\frac{K}{2g}$ complex symbol encodable) if $g$ divides $K$ and if $\mathscr{A}=\mathscr{A}_1\times\mathscr{A}_2\times\dots\times\mathscr{A}_g$ where each $\mathscr{A}_i,i=1,\dots,g\subset\mathbb{R}^{\frac{K}{g}}$.
\end{defn} 

Suppose we partition the set of weight matrices of $S(X)$ into $g$-groups, the $k$-th group containing $K/g$ matrices and also the information symbol vector as 
$$X=\left[\begin{array}{cccc}X_1^T X_2^T \dots X_g^T\end{array}\right]^T$$
where 
$$X_k=\left[\begin{array}{cccc}x_{\frac{(k-1)K}{g}+1} & x_{\frac{(k-1)K}{g}+1} & \dots & x_{\frac{kK}{g}}\end{array}\right]^T.$$ 
Now $S(X)$ can be written as,
$$
S(X)=\sum_{k=1}^{g}S_k(X_k); \quad S_k(X_k)=\sum_{i=\frac{(k-1)K}{g}+1}^{\frac{kK}{g}}x_iB_i.
$$
Minimizing the decoding metric corresponding to \eqref{eqn_metric0}
\begin{equation}
\label{eqn_metric}
\parallel y_t-\frac{1}{a_{t-1}}S(X)y_{t-1}\parallel^2
\end{equation}
is in general not same as minimizing
\begin{equation}
\label{eqn_submetric}
\parallel y_t-\frac{1}{a_{t-1}}S_k(X_k)y_{t-1}\parallel^2
\end{equation}
for each $1\leq k\leq g$ individually. However, if it so happens then the decoding complexity is reduced by a large amount. Note that it is not possible to compute \eqref{eqn_submetric} unless the code is $g$-group encodable also.

\begin{defn}
\label{defn2}
A linear STBC $\mathscr{C}=\left\{S(X)|X\in \mathscr{A}\right\}$ is said to be $g$-group decodable (or $\frac{K}{g}$ real symbol decodable or $\frac{K}{2g}$ complex symbol decodable) if it is $g$-group encodable and if its decoding metric in \eqref{eqn_metric} can be simplified as in \eqref{eqn_submetric}.
\end{defn}

The DDSTC design problem is then to design a $g$-group decodable linear STBC 
$$\mathscr{C}=\left\{S(X=\left[\begin{array}{cccc}x_1 & x_2 & \dots & x_K\end{array}\right])|X\in\mathscr{A}\right\}$$
such that
\begin{enumerate}
\item All codewords are scaled unitary matrices respecting the transmit power constraint
\item $K$ and $g$ are maximized.
\item \label{cond_DDSTC4} There exist $R$ unitary matrices $A_1,A_2,\dots,A_R$ of size $T\times T$ such that the first $M$ of them satisfy $A_iC=CA_i,\ i=1,\dots,M,\ \forall\ C\in\mathscr{C}$ and the remaining $R-M$ of them satisfy $A_iC^*=CA_i,\ i=M+1,\dots,R,\ \forall\ C\in\mathscr{C}$.
\item \label{cond_DDSTC5} There exists an initial vector $s_0$ such that the initial matrix $X_0=\left[\begin{array}{ccccccc}A_1s_0 & A_2s_0 & \dots & A_Ms_0 & A_{M+1}s_0^* & \dots & A_Rs_0^*\end{array}\right]$ is unitary.
\item $\min_{S_1,S_2\in\mathscr{C}}|S_1-S_2|$ is maximized.
\end{enumerate}

Observe that the requirements for designing DDSTCs are more restrictive than that for DSTCs. Note that condition \ref{cond_DDSTC4} and condition \ref{cond_DDSTC5} were not required for designing DSTCs. As an additional requirement it would be nice to have a single design $S(x_1,x_2,\dots,x_K)$ and a family of signal sets, one for each transmission rate such that all the required conditions are met. This means that we need to be able to find $R$ relay matrices satisfying the required conditions irrespective of the size of the code $|\mathscr{C}|$.
\subsection{DSTCs from extended Clifford algebras}
\label{sec3a}
In this subsection, we briefly describe the constructions of a class of rate one, linear designs satisfying the conditions for four-group decodability which were obtained using extended Clifford algebras in \cite{RaRisit3}. This algebraic framework is needed for us to be able to prove that the conditions 3) and 4) for DDSTCs are satisfied by codes arising out of these linear designs. Signal sets which lead to full diversity for these linear designs are provided in \cite{RaRisit1}.

\begin{defn}
Let $L=2^a,a\in \mathbb{N}$. An Extended Clifford algebra denoted by $\mathbb{A}_n^L$ is the associative algebra over $\mathbb{R}$~ generated by $n+a$ objects $\gamma_k,\ k=1,\dots,n$ and $\delta_i,\ i=1,\dots,a$ which satisfy the following relations:
\begin{itemize}
\item $\gamma_k^2=-1,\ \forall\ k=1,\dots,n$
\item $\gamma_k\gamma_j=-\gamma_j\gamma_k,\ \forall\ k\neq j$
\item $\delta_k^2=1,\ \forall k=1,\dots,a$
\item $\delta_k\delta_j=\delta_j\delta_k,\ \forall\ 1\leq k,j\leq a$
\item $\delta_k\gamma_j=\gamma_j\delta_k,\ \forall\ 1\leq k\leq a, 1\leq j\leq n$
\end{itemize}
\end{defn}

It is clear that the classical Clifford algebra, denoted by $Cliff_n$, is obtained when only the first two relations are satisfied and there are no $\delta_i$. $Cliff_n$ is a sub-algebra of $\mathbb{A}_n^L$. Let $\mathscr{B}_n$ be the natural $\mathbb{R}$-basis for this sub-algebra. Then a natural $\mathbb{R}$-basis for $\mathbb{A}_n^L$ is
\begin{equation}
\begin{array}{rl}
\mathscr{B}_n^L=&\mathscr{B}_n\cup\left\{\mathscr{B}_n\delta_i|i=1,\dots,a\right\}\\
&\bigcup_{m=2}^{a}\mathscr{B}_n\left\{\prod_{i=1}^{m}\delta_{k_i}|1\leq k_i\leq k_{i+1}\leq a\right\}
\end{array}
\end{equation}
\noindent
where
\begin{equation}
\begin{array}{rl}
\mathscr{B}_n=&\left\{1\right\}\bigcup\left\{\gamma_i|i=1,\dots,n\right\}\\
&\bigcup_{m=2}^{n}\left\{\prod_{i=1}^{m}\gamma_{k_i}|1\leq k_i\leq k_{i+1}\leq n\right\}
\end{array}
\end{equation}
A unitary matrix representation for the symbols $1,\gamma_1,\gamma_2,\gamma_1\gamma_2,\delta_k,k=1,\dots,a,\bigcup_{m=2}^{a}\prod_{i=1}^{m}\delta_{k_i}|1\leq k_i\leq k_{i+1}\leq a$ in the algebra $\mathbb{A}_2^L$ is needed \cite{RaRisit3} to construct linear designs which are four-group decodable. We briefly explain the matrix representation procedure and then illustrate it with few example.

We first view $\mathbb{A}_2^L$ as a vector space over $\mathbb{C}$ by thinking of $\gamma_1$ as the complex number $i=\sqrt{-1}$. A natural $\mathbb{C}$-basis for $\mathbb{A}_2^L$ is given by
\begin{equation}
\begin{array}{rl}
\mathcal{B}_n^L=&\left\{1,\gamma_2\right\}\cup\left\{\left\{1,\gamma_2\right\}\delta_i|i=1,\dots,a\right\}\\
& \bigcup_{m=2}^{a}\left\{1,\gamma_2\right\}\left\{\prod_{i=1}^{m}\delta_{k_i}|1\leq k_i\leq k_{i+1}\leq a\right\}.
\end{array}
\end{equation}
Thus the dimension of $\mathbb{A}_2^L$ seen as a vector space over $\mathbb{C}$ is $2^{n+a-1}$.

We have a natural embedding of  $\mathbb{A}_2^L$ into $\mathrm{End}_{\mathbb{C}}(\mathbb{A}_2^L),$ (the set of all ${\mathbb C}$-linear maps from $\mathbb{A}_2^L$ to itself) given by left multiplication as shown below.

\begin{equation*}
\begin{array}{l}
\phi:\mathbb{A}_2^L\mapsto \mathrm{End}_{\mathbb{C}}(\mathbb{A}_2^L) \\
\phi(x)=L_x:y\mapsto xy.
\end{array}
\end{equation*}

Since $L_x$ is $\mathbb{C}$-linear, we can write down a matrix representation of $L_x$ with respect to the natural $\mathbb{C}$-basis $\mathcal{B}_n^L$. Left regular representation naturally yields unitary matrix representations for the required symbols in the algebra. The resulting designs are $4$-group decodable \cite{RaRisit3}.

The design for $4$ relays  is
\begin{equation} 
S=\left[\begin{array}{ccrr}
s_1 & s_2 & -s_3^* & -s_4^*\\ 
s_2 & s_1 & -s_4^* & -s_3^*\\
s_3 & s_4 & s_1^* & s_2^*\\
s_4 & s_3 & s_2^* & s_1^*
\end{array}
\right] 
\end{equation} 
and the design for $8$ relays is 
\begin{equation} \label{eqn_design8} \left[\begin{array}{ccccrrrr}
s_1 & s_2 & s_3 & s_4 & -s_5^* & -s_6^* & -s_7^* & -s_8^*\\
s_2 & s_1 & s_4 & s_3 & -s_6^* & -s_5^* & -s_8^* & -s_7^*\\
s_3 & s_4 & s_1 & s_2 & -s_7^* & -s_8^* & -s_5^* & -s_6^*\\
s_4 & s_3 & s_2 & s_1 & -s_8^* & -s_7^* & -s_6^* & -s_5^*\\
s_5 & s_6 & s_7 & s_8 & s_1^* & s_2^* & s_3^* & s_4^*\\
s_6 & s_5 & s_8 & s_7 & s_2^* & s_1^* & s_4^* & s_3^*\\
s_7 & s_8 & s_5 & s_6 & s_3^* & s_4^* & s_1^* & s_2^*\\
s_8 & s_7 & s_6 & s_5 & s_4^* & s_3^* & s_2^* & s_1^* \end{array}\right]
\end{equation} 
and for the 16 relays  we get the $16\times 16$ design shown in \eqref{eqn_16} at the top of the next page. The partitioning of the real variables of the design into four groups is provided in \cite{RaRisit1}.

\begin{figure*}
{\normalsize
\begin{equation}
\label{eqn_16}
S=\left[\begin{array}{ccccccccrrrrrrrr}
s_1 & s_2 & s_3 & s_4 & s_5 & s_6 & s_7 & s_8 & -s_{9}^* & -s_{10}^* & -s_{11}^* & -s_{12}^* & -s_{13}^* & -s_{14}^* & -s_{15}^* & -s_{16}^*\\
s_2 & s_1 & s_4 & s_3 & s_6 & s_5 & s_8 & s_7 & -s_{10}^* & -s_{9}^* & -s_{12}^* & -s_{11}^* & -s_{14}^* & -s_{13}^* & -s_{16}^* & -s_{15}^*\\
s_3 & s_4 & s_1 & s_2 & s_7 & s_8 & s_5 & s_6 & -s_{11}^* & -s_{12}^* & -s_{9}^* & -s_{10}^* & -s_{15}^* & -s_{16}^* & -s_{13}^* & -s_{14}^*\\
s_4 & s_3 & s_2 & s_1 & s_8 & s_7 & s_6 & s_5 & -s_{12}^* & -s_{11}^* & -s_{10}^* & -s_{9}^* & -s_{16}^* & -s_{15}^* & -s_{14}^* & -s_{13}^*\\
s_5 & s_6 & s_7 & s_8 & s_1 & s_2 & s_3 & s_4 & -s_{13}^* & -s_{14}^* & -s_{15}^* & -s_{16}^* & -s_{9}^* & -s_{10}^* & -s_{11}^* & -s_{12}^*\\
s_6 & s_5 & s_8 & s_7 & s_2 & s_1 & s_4 & s_3 & -s_{14}^* & -s_{13}^* & -s_{16}^* & -s_{15}^* & -s_{10}^* & -s_{9}^* & -s_{12}^* & -s_{11}^*\\
s_7 & s_8 & s_5 & s_6 & s_3 & s_4 & s_1 & s_2 & -s_{15}^* & -s_{16}^* & -s_{13}^* & -s_{14}^* & -s_{11}^* & -s_{12}^* & -s_{9}^* & -s_{10}^*\\
s_8 & s_7 & s_6 & s_5 & s_4 & s_3 & s_2 & s_1 & -s_{16}^* & -s_{15}^* & -s_{14}^* & -s_{13}^* & -s_{12}^* & -s_{11}^* & -s_{10}^* & -s_{9}^*\\
s_9 & s_{10} & s_{11} & s_{12} & s_{13} & s_{14} & s_{15} & s_{16} & s_1^* & s_2^* & s_3^* & s_4^* & s_5^* & s_6^* & s_7^* & s_8^*\\
s_{10} & s_9 & s_{12} & s_{11} & s_{14} & s_{13} & s_{16} & s_{15} & s_2^* & s_1^* & s_4^* & s_3^* & s_6^* & s_5^* & s_8^* & s_7^*\\
s_{11} & s_{12} & s_9 & s_{10} & s_{15} & s_{16} & s_{13} & s_{14} & s_3^* & s_4^* & s_1^* & s_2^* & s_7^* & s_8^* & s_5^* & s_6^*\\
s_{12} & s_{11} & s_{10} & s_9 & s_{16} & s_{15} & s_{14} & s_{13} & s_4^* & s_3^* & s_2^* & s_1^* & s_8^* & s_7^* & s_6^* & s_5^*\\
s_{13} & s_{14} & s_{15} & s_{16} & s_9 & s_{10} & s_{11} & s_{12} & s_5^* & s_6^* & s_7^* & s_8^* & s_1^* & s_2^* & s_3^* & s_4^*\\
s_{14} & s_{13} & s_{16} & s_{15} & s_{10} & s_9 & s_{12} & s_{11} & s_6^* & s_5^* & s_8^* & s_7^* & s_2^* & s_1^* & s_4^* & s_3^*\\
s_{15} & s_{16} & s_{13} & s_{14} & s_{11} & s_{12} & s_9 & s_{10} & s_7^* & s_8^* & s_5^* & s_6^* & s_3^* & s_4^* & s_1^* & s_2^* \\
s_{16} & s_{15} & s_{14} & s_{13} & s_{12} & s_{11} & s_{10} & s_9 & s_8^* & s_7^* & s_6^* & s_5^* & s_4^* & s_3^* & s_2^* & s_1^*
\end{array}\right]
\end{equation}
}
\hrule
\end{figure*}
\section{Explicit construction of DDSTCs}
\label{sec4}
In this section, we shall prove that the additional requirements are met by the constructed codes in \cite{RaRisit1}. 

\begin{thm} 
The extra conditions (conditions \ref{cond_DDSTC4}) and \ref{cond_DDSTC5}) of the DDSTC design problem stated immediately after Definition \ref{defn2}) are met by the designs from extended Clifford algebras.
\end{thm}
\begin{proof} 
We prove the existence of the $R$ relay matrices by explicitly constructing them. For this purpose, we shall use the fact that the code for $R=2^\lambda$ relays was obtained as a matrix representation of the Extended Clifford algebra $\mathbb{A}_2^{2^{\lambda-1}}$. Thus $a=\lambda-1$. We choose $M=2^a=\frac{R}{2}$. The $M$ relay matrices are given by the union of the elements of the sets

{\small
$$\left\{\phi(1),\phi(\delta_1),\dots,\phi(\delta_a)\right\},~\left\{\bigcup_{m=2}^{a}\prod_{i=1}^{m}\phi(\delta_{k_i})|1\leq k_i\leq k_{i+1}\leq a\right\}.$$ 
}

By virtue of the property that the map $\phi$ is a ring homomorphism, these matrices are guaranteed to commute with all the codewords because they are matrix representations of elements belonging to the center of the algebra $\mathbb{A}_2^{2^{\lambda-1}}$. To obtain the remaining $R-M$ relay matrices, we need to find unitary matrices which satisfy
\begin{equation}
\label{eqn_skewcommute}
A_iC^*=CA_i,i=M+1,\dots,R
\end{equation}
where $C$ is any codeword. But the codeword $C$ is a matrix representation of some element belonging to the Extended Clifford algebra. One method to get these relay matrices is to take them from within the Extended Clifford algebra itself. By doing so, we can translate the condition in \eqref{eqn_skewcommute} into a condition on elements of the algebra which will provide us a handle on the problem. Towards that end, we first need to identify a map in the algebra which is the analogue of taking the conjugate of the  matrix representation of an element. Recall that in Subsection \ref{sec3a}, we used the fact that $\gamma_1$ can be thought of as the complex number $i=\sqrt{-1}$. Note that when we take the conjugate of a matrix, we  simply replace $i$ by $-i$. Hence the analogue of this action in the algebra is to replace $\gamma_1$ by $-\gamma_1$. Thus, we define the analogous map $\sigma$ in the algebra as follows:
\begin{equation}
\sigma:x\mapsto\bar{x}
\end{equation}
where, the element $\bar{x}$ is obtained from $x$ by simply replacing $\gamma_1$ by $-\gamma_1$ in the expression of $x$ in terms of the natural ${\mathbb R}$-basis of Extended Clifford algebra. Now the problem is to find $R-M$ distinct elements denoted by $a_i,i=M+1,\dots,R$ of the algebra $\mathbb{A}_2^{2^{\lambda-1}}$ which satisfy
\begin{equation}
\label{eqn19}
a_i\bar{x}=xa_i,\ \forall x\in\mathbb{A}_2^{2^{\lambda-1}}.
\end{equation}
The elements 
{\small
$$\left\{\gamma_2\left\{1,\delta_1,\dots,\delta_a\right\}\right\},~ \left\{\gamma_2\left\{\bigcup_{m=2}^{a}\prod_{i=1}^{m}\phi(\delta_{k_i})|1\leq k_i\leq k_{i+1}\leq a\right\}\right\}$$ 
}

\noindent
satisfy the above required condition. This can be easily proved by using the fact that $\gamma_2(-\gamma_1)=\gamma_1(\gamma_2)$(anti-commuting property). Hence the required unitary relay matrices $A_i,i=M+1,\dots,R$ can be obtained by taking the matrix representation of these specific elements. If we plug in these relay matrices to form a design we get 
$$X=\left[\begin{array}{cccccc}A_1s & \dots & A_Ms & A_{M+1}s^* & \dots & A_Rs^*\end{array}\right],$$ 
where $s=\left[\begin{array}{cccc}x_1 & x_2 & \dots & x_R\end{array}\right]^T$ is the vector of complex information symbols. Then, we get exactly the same design which is used at the source. Because of this we can choose the initial vector $s_0=\left[\begin{array}{cccc}1 & 0 & \dots & 0\end{array}\right]^T$ which then guarantees that the initial matrix is also unitary.
\end{proof}

However, we would like to point out that there are also other elements of the algebra which satisfy the requirements \eqref{eqn19}. For example consider the union of the elements of the sets

{\footnotesize
$$\left\{\gamma_1\gamma_2\left\{1,\delta_1,\dots,\delta_a\right\}\right\},~ \left\{\gamma_1\gamma_2\left\{\bigcup_{m=2}^{a}\prod_{i=1}^{m}\phi(\delta_{k_i})|1\leq k_i\leq k_{i+1}\leq a\right\}\right\}.$$
}

\begin{eg}
Let $R=4$. Then the DDSTC $\mathscr{C}$ is obtained using the design
$$
\left[\begin{array}{rrrr}
s_1 & s_2 & -s_3^* & -s_4^*\\
s_2 & s_1 & -s_4^* & -s_3^*\\
s_3 & s_4 & s_1^* & s_2^*\\
s_4 & s_3 & s_2^* & s_1^*
\end{array}\right]
$$
and the signal set constructed in \cite{RaRisit1}. The signal set is a Cartesian product of four $2$-dimensional signal sets. The relay matrices are given as follows:
$$
\begin{array}{rcc}
A_1&=&I_4\\
A_2&=&\left[\begin{array}{cccc}0 & 1 & 0 & 0\\
1 & 0 & 0 & 0\\
0 & 0 & 0 & 1\\
0 & 0 & 1 & 0\end{array}\right]\\
A_3&=&\left[\begin{array}{ccrr}0 & 0 & -1 & 0\\
0 & 0 & 0 & -1\\
1 & 0 & 0 & 0\\
0 & 1 & 0 & 0\end{array}\right]\\
A_4&=&\left[\begin{array}{ccrr}0 & 0 & 0 & -1\\
0 & 0 & -1 & 0\\
0 & 1 & 0 & 0\\
1 & 0 & 0 & 0\end{array}\right]
\end{array}
$$
The initial vector $s_0=\left[\begin{array}{cccc}1 & 0 & \dots & 0\end{array}\right]^T$ and thus the initial matrix $X_0=I_4$. This DDSTC is single complex decodable (or 2 real symbol decodable).
\end{eg}

\section{Discussion}
In this paper, we have constructed a class of four group decodable DDSTCs for $R=2^\lambda$ relays using extended Clifford algebras. Note that relaxing the unitary matrix codebook to scaled unitary matrix codebook has paved us the way to obtain the decoding complexity benefits. It would be interesting to know whether there are unitary matrix codebooks with the same decoding complexity, since then the decoder in \eqref{eqn_metric0} would correspond to an ML decoder. Maximizing the coding gain is also an interesting problem for further research.
\section*{Acknowledgment}
This work was supported through grants to B.S.~Rajan; partly by the
IISc-DRDO program on Advanced Research in Mathematical Engineering, and partly
by the Council of Scientific \& Industrial Research (CSIR, India) Research
Grant (22(0365)/04/EMR-II). The authors would like to thank Dr.Frederique Oggier and Dr.Yindi Jing for sending us the preprints of their recent works \cite{OgH1,JiJ,OgH2}. 


\begin{thebibliography}{1}














\bibitem{OgH1} F. Oggier, B. Hassibi. Algebraic Cayley differential Space-Time Codes , \emph{IEEE Transactions on Information Theory}, Vol. 53, No. 5, May 2007. Available online http://www.systems.caltech.edu/\~{}frederique/draftcayley.ps






\bibitem{TaC} M. Tao and R. S. Cheng, ``Differential space-time block codes," Proceedings of \emph{IEEE Globecom 2001}, Vol. 2, pp. 1098-1102, San Antonio, USA, 25-29, Nov. 2001.

\bibitem{ZhJ} Yun Zhu and Hamid Jafarkhani, ``Differential Modulation Based on
Quasi-Orthogonal Codes," \emph{IEEE Transactions on Wireless Communications}, Vol. 4, No. 6, pp. 3018-3030, Nov. 2005.


\bibitem{YGT2} Chau Yuen, Yong Liang Guan, Tjeng Thiang Tjhung, ``Single-symbol decodable differential space-time modulation based on QO-STBC," to appear in \emph{IEEE Transactions on Wireless Communications}. Available in arxiv cs.IT/0605095.

\bibitem{JiH4} Yindi Jing and B. Hassibi, ``Distributed space-time coding in wireless relay networks," \emph{IEEE Transactions on Wireless Communications}, Vol.5, No.12, pp.3524-3536, Dec. 2006. Available online http://www.ee.caltech.edu/EE/Faculty/babak/pubs/papers/dstc.pdf


\bibitem{KiR2} Kiran T. and B. Sundar Rajan, ``Partially-coherent distributed space-time codes with differential encoder and decoder,'' IEEE Journal on Selected Areas in Communications: Special issue on Cooperative Communications and Networking, Vol.25, No.2, Feb. 2007, pp.426-433.

\bibitem{JiJ} Y. Jing and H. Jafarkhani,``Distributed Differential Space-Time Coding for Wireless Relay Networks," submitted to \emph{IEEE Transactions on Communications}. Private Communication.

\bibitem{OgH2}  Fr\'{e}d\'{e}rique Oggier, Babak Hassibi, "A Coding Strategy for Wireless Networks with no Channel Informations", Presented at the \emph{Forty-Fourth Annual Allerton Conference on Communication, Control and Computing}, Sep. 27-29, 2006.





\bibitem{RaR1} G. Susinder Rajan and B.Sundar Rajan, ``A Non-orthogonal distributed space-time protocol, Part-I: Signal model and design criteria,'' Proceedings of \emph{IEEE International Workshop in Information Theory}, Chengdu, China, Oct.22-26, 2006, pp. 385-389.

\bibitem{RaRisit3} G. Susinder Rajan and B.Sundar Rajan, ``Algebraic distributed space-time codes with low ML decoding complexity,'' to appear in Proceedings of \emph{IEEE ISIT 2007}, Nice, France.

\bibitem{RaRisit1} G. Susinder Rajan and B.Sundar Rajan, ``Signal set design for full-diversity low-decoding complexity differential scaled-unitary STBCs,'' to appear in Proceedings of \emph{IEEE ISIT 2007}, Nice, France.





















\end{thebibliography}
\end{document}